\documentclass[journal]{IEEEtran}
\usepackage{amssymb}
\usepackage[cmex10]{amsmath}
\usepackage{stfloats}
\usepackage{graphicx}
\usepackage{subfigure}
\usepackage{tabularx}
\usepackage{epsfig,epsf,color,balance,cite}
\usepackage{verbatim}
\usepackage{url}
\usepackage{bm}
\usepackage{amsthm}
\newtheorem{theorem}{Theorem}

\usepackage{algorithm}
\usepackage{algorithmic}

\usepackage{color}
\definecolor{myc1}{rgb}{0,0,0}

\begin{document}
	
	\title{Energy Efficient Semantic Communication over Wireless Networks with Rate Splitting}
	\author{
		\IEEEauthorblockN{Zhaohui Yang,
			Mingzhe Chen,
			Zhaoyang~Zhang,
			and Chongwen Huang
		}
	\thanks{This work was supported in part by National Natural Science Foundation
		of China under Grant 61725104 and U20A20158, and National Key R\&D
		Program of China under Grant 2018YFB1801104 and 2020YFB1807101. The work of Prof. Huang was supported by the China National Key R\&D Program under Grant 2021YFA1000500, National Natural Science Foundation of China under Grant 62101492, Zhejiang Provincial Natural Science Foundation of China under Grant LR22F010002, National Natural Science Fund for Excellent Young Scientists Fund Program (Overseas), Zhejiang University Education Foundation Qizhen Scholar Foundation, and Fundamental Research Funds for the Central Universities under Grant 2021FZZX001-21.  }
		\thanks{Z. Yang is  with Zhejiang Lab,  Hangzhou, Zhejiang, 311121, China.  Z. Yang, Z. Zhang, and C. Huang are with the College of Information Science and Electronic Engineering, Zhejiang University, Hangzhou, Zhejiang 310027, China, and Zhejiang Provincial Key Lab of Information Processing, Communication and Networking (IPCAN), Hangzhou, Zhejiang, 310007, China. (e-mails: yang\_zhaohui@zju.edu.cn, ning\_ming@zju.edu.cn, chongwenhuang@zju.edu.cn) }
		\thanks{M. Chen is with the Department of Electrical and Computer Engineering and Institute for Data Science and Computing, University of Miami, Coral Gables, FL, 33146 USA  (e-mail:mingzhe.chen@miami.edu)}
	}
	

		\maketitle
		\vspace{-0.25em}
		\begin{abstract}
			In this paper, the problem of wireless resource allocation and semantic information extraction for energy efficient semantic communications over wireless networks with rate splitting is investigated. In the considered model, a base station (BS) first extracts semantic information from its large-scale data, and then transmits the small-sized semantic information to each user which recovers the original data based on its local common knowledge. At the BS side, the probability graph is used to extract multi-level semantic information. In the downlink transmission, a rate splitting scheme is adopted, while the private small-sized semantic information is transmitted through private message and the common knowledge is transmitted through common message. Due to limited wireless resource, both computation energy and transmission energy are considered. This joint computation and communication problem is formulated as an optimization problem aiming to minimize the total communication and computation energy consumption of the network under computation, latency, and transmit power constraints. To solve this problem, an alternating algorithm is proposed where the closed-form solutions for semantic information extraction ratio  and computation frequency are obtained at each step. Numerical results verify the effectiveness of the proposed algorithm.
		\end{abstract}
		\IEEEpeerreviewmaketitle
		
		\section{Introduction}
		The rapid development of emerging applications such as digital twin, edge learning, and metaverse requires wireless networks to support high transmission data rate, ultra low latency, and seamless connectivity \cite{mao2022rate,clerckx2022primer,xu2022edge,9562559}. However, due to limited wireless resources such as frequency and time, conventional orthogonal multiple access schemes  cannot support massive connectivity concern for next-generation wireless communication networks \cite{saad2019vision}.  
		Through using the same time or frequency resource, multiple users can be served in non-orthogonal multiple access (NOMA) \cite{liu2018non,vaezi2018multiple,7263349,Zhiguo2017Survey,8422457,saito2013non}, where users can be split  in the power or code domain.
		Since additional users can be served with superposition coding at the transmitter and successive interference cancellation (SIC) at the receiver, the spectral efficiency of NOMA is generally higher than  conventional orthogonal multiple access schemes. 
		
		In downlink NOMA transmission,  the receiver side decodes the interference for all received strong messages  \cite{saito2013non}. Thus, the computation capacity of NOMA decoding is generally high.  
		To balance the decoding tradeoff  of intended signal and interference signal, the concept of rate splitting multiple access (RSMA)  was introduced in  \cite{485709,1056307,7470942,liu2020rate}. 
		For downlink RSMA transmission, the transmission message intended for each user is divided into both common and private parts. 
		All users intend to receive the common part of the message, i.e., common message, while only part of the users wish to receive and decode the specific private part of the message, i.e., private message. 
		At the user side, the common message is decoded first with regarding all private messages as interference, while the intended private message is decoded with only considering the private messages intended for other users as interference. 
		Through dynamically controlling the split of private and common messages, the computation complexity of RSMA can be adjusted to achieve the specific spectral efficiency requirements. 
		To implement RSMA for wireless communication systems, there are still many challenges, which include the resource allocation for private and common messages, decoding order optimization, system design in imperfect channel and hardware mismatch cases.
		
		There are many contributions investigating the problems of RSMA in wireless communication systems.
		The general challenges of RSMA were pointed out in \cite{7470942} for multiple input multiple output (MIMO) communication systems.
		To maximize the sum rate of all users, a distributed rate splitting technique was proposed in \cite{4039650}. 
		For a two-receiver  multiple input single output (MISO) communication system with limited rate feedback, the rate analysis was investigated in  \cite{7152864}. 
		Compared with NOMA and space-division multiple access (SDMA), it was shown in \cite{mao2018rate} that RSMA can achieve the best performance in terms of spectral and energy efficiency \cite{li2022rate}. 
		In particular, the energy efficiency optimization for RSMA and NOMA transmissions in a unmanned aerial vehicle assisted wireless communication system was investigated in \cite{rahmati2019energy} .
		Considering wireless energy transfer and information transmission,   the linear precoding method  for RSMA was investigated in \cite{mao2019rate}. 
		For the case with imperfect channel state information,  the sum rate maximization with partial channel state information for RSMA was studied in 
		\cite{7555358}, while  a downlink MISO RSMA system with bounded channel errors was  investigated in \cite{7513415}.
		
			{\color{myc1}The interplay between rate splitting with emerging technologies has been investigated. 
		With the help of reconfigurable intelligent surface,  the energy efficient resource allocation for reconfigurable intelligent surface assisted RSMA was investigated in
		\cite{yang2020energy},  where the phase shift, rate allocation, and trasnmit beamforming were jointly scheduled. 
		The learning based traffic prediction method was studied in  \cite{lu2020machine} for unmanned aerial vehicle enabled wireless communication system with rate splitting. 
		The neural network was proposed in \cite{pereira2022user} to solve the user clustering problem in hierarchical rate splitting communication systems.
		Due to coupled rate and power allocation relationship, the resource allocation of RSMA usually leads to the nonconvex problem, which can be solved by utilizing the learning techniques such as deep reinforcement learning. 
		Several deep learning algorithms were designed to solve various complex resource allocation problems for RSMA, which include total power minimization problem
		\cite{camana2022deep},  joint power control, beamforming design, and splitting optimization problem  \cite{liu2022joint,yang2020sum}, power allocation problem with limited channel state information knowledge \cite{hieu2021optimal,huang2022deep}, 	
		joint transmit power, user clustering, and resource block allocation problem \cite{hassan2021joint}, joint passive precoding at the reconfigurable intelligent surface and active precoding at the transmitter \cite{wu2022deep}.
		In the federated learning frameworks \cite{chen2021a}, the authors in 
		\cite{park2022completion} utilized RSMA for uplink model transmission to minimize the total delay of the whole system.
		A model-based deep  learning algorithm was developed to solve the receiver design problem of RSMA in \cite{loli2022model}. 
		
	   Recently, semantic communication has attracted a lot of attention \cite{weng2021semantic,wang2022performance,xie2021deep,chen2022performance,kang2022personalized,yang2023rate,kaewpuang2022cooperative,xie2022robust,wang2022wireless,zou2022goal}.
		For the wireless communication system characterized by Shannon capacity, the receiver side needs to recover the information that is exactly the same as the transmitted information. 
		However, in the emerging wireless applications such as virtual reality, personalized healthcare, autonomous driving, and the Internet-of-Everything (IoE), the wireless communication systems aim to meet the  multimodal quality-of-experience (QoE) requirements with massive data, which makes the traditional Shannon capacity characterized transmission infeasible.   
		Especially in human-computer interaction scenarios, humans can control multiple IoE devices simultaneously through voice and augmented/virtual reality commands, making communication ubiquitous in small-range wireless networks, which poses severe challenges to traditional bit-oriented communication challenge. Supporting real-time human-machine interaction and machine-to-machine interaction through the use of text, speech, images, and augmented/virtual reality is important for future wireless communications. In order to support this interaction, the important information finally received depends mainly on the intent, rather than the bit information dependence of common sense. These applications use advanced signal processing to facilitate the development of task-oriented semantic communication \cite{deniz2022beyond,qin2021semantic,xu2022edge}. In semantic communication, both transmitter and receiver share common knowledge, which can be used to extract small-size information at the transmitter and recover the original information at the receiver \cite{tong2021federated}. Similarly, in downlink RSMA, all users also need to receive both common information and private information. Due to the inherent similarity of common  knowledge and common message, RSMA can be utilized to enhance the system performance of downlink semantic communication. To our best knowledge, there is no prior works that consider the integration of  semantic communication and RSMA. }

		The main contributions of this paper
		include:
		\begin{itemize}
			\item The problem of wireless resource allocation and semantic information extraction for energy efficient semantic communications over wireless networks with rate splitting is investigated. In the considered model, the BS first extracts the semantic information from its large-scale data, and then transmits the small-sized semantic information to each user which recovers the original data based on the local common knowledge. 
			\item In the downlink transmission, the rate splitting scheme is adopted, while the private small-sized semantic information is transmitted through private message and the common knowledge is transmitted through common message. Due to limited wireless resource, both computational energy and transmission energy must be considered. This joint computation and communication problem is formulated as an optimization problem whose goal is to minimize the total energy consumption of the network under a latency constraint. 
			\item To solve this problem, an iterative algorithm is proposed where, at every step, closed-form solutions for semantic information extraction ratio  and computation frequency are derived. Numerical results show the effectiveness of the proposed algorithm.
		\end{itemize}
		The rest of this paper is organized as follows. The system model and problem formulation are described in Section \uppercase\expandafter{\romannumeral2}. The algorithm design is presented in Section \uppercase\expandafter{\romannumeral3}. Simulation results are analyzed in Section \uppercase\expandafter{\romannumeral4}. Conclusions are drawn in Section \uppercase\expandafter{\romannumeral5}.

		\section{System Model and Problem Formulation}
		Consider a downlink semantic wireless communication (SWC) network with one multiple-antenna BS and $K$ single-antenna users, as shown in Fig. \ref{sys1}. The BS is equipped with $N$ antennas and the set of users is denoted by $\mathcal K$.
		Each user $k$ has a large-sized data $\mathcal D_k$ to receive.
		Due to limited wireless resource, the BS needs to extract the small-sized semantic information from the original data $\mathcal D_k$.
		In the considered model, the BS first extracts the semantic information based on directional probability graph and then transmits the semantic information via rate splitting technique.  
		
		\begin{figure}
			\centering
			\includegraphics[width=3.5in]{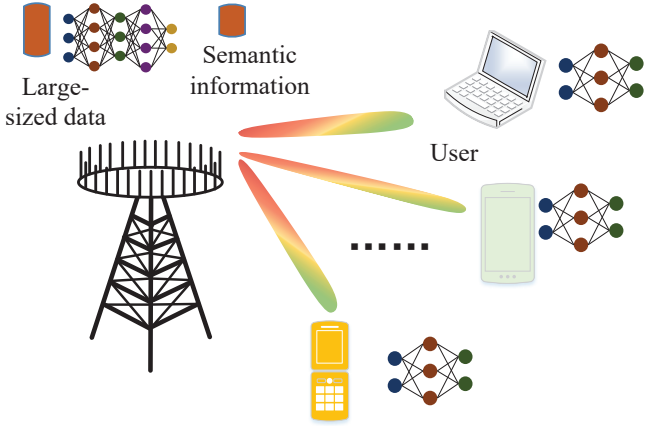}
			\vspace{-.5em}
			\caption{Illustration of the considered SWC network with rate splitting.}
			\vspace{-.5em}
			\label{sys1}
		\end{figure}

		\subsection{Semantic Communication Model}
		In this part, we utilize the directional probability graph to characterize the inherent information of the transmitted information.  In the directional probability graph, each vertex represents the semantic entity with different semantic levels. The higher level the semantic level is, the more complicated  the semantic information is. The link between any two vertexes represents the probability.
		
		To construct the directional probability graph, we use the deep neural network to train the  stored dataset, which includes three main steps.
		In the first step,  {\color{myc1}the semantic entity is recognized from the dataset, where the semantic entity means the names in text, including person names, place names, etc. The name of semantic entity is highly open (various types, flexible lengths, unregistered words), contains rich knowledge and highlights individuality}. Three common methods, i.e., 
		rule method, taxonomy method, and sequence labeling \cite{deng2018deep} can be used to identify the semantic entity. 
		The  semantic entity is presented as a vertex in the directional probability graph. 
		In the second step, the link between any two vertexes means the probability that one vertex can be linked with the other vertex. 
		Through training the dataset, the probability between two vertexes can be calculated via convolutional neural networks.   
		In the third step, the semantic information  fusion is conducted. For two vertexes, if the link probabilities between these two vertexes are higher than a predefined threshold. As a result, the final directional probability graph becomes a multi-tier graph, as shown in Fig. \ref{sys3}.

		\begin{figure}
			\centering
			\includegraphics[width=3.5in]{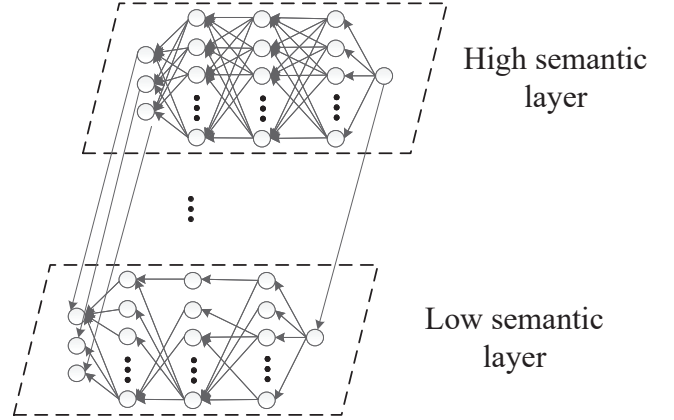}
			\vspace{-.5em}
			\caption{An example of the multi-level semantic information extraction.}
			\vspace{-.5em}
			\label{sys3}
		\end{figure}
		
		To obtain the small-size semantic communication, the extraction process includes two parts, as shown in Fig.~\ref{sys0}. In the first part, the directional probability graph is used to extract semantic information and the output is denoted by $\mathcal G(\mathcal D_k)$.
		To efficiently transmit information, in the second part, a subset $\mathcal S_k$ out of $\mathcal G(\mathcal D_k)$ is selected at user $k$, which is used for data transmission.
		
		\begin{figure}
			\centering
			\includegraphics[width=3.5in]{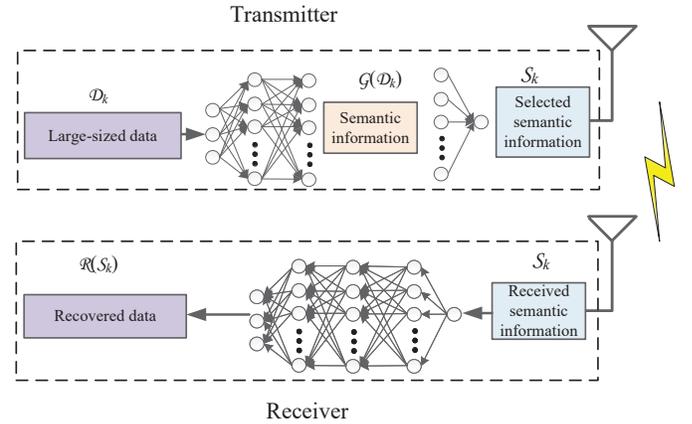}
			\vspace{-.5em}
			\caption{Illustration of the SWC model.}
			\vspace{-.5em}
			\label{sys0}
		\end{figure}

		At the user side, each user utilizes the shared common directional probability graph to recover the original data and the recovered data is denoted by $\mathcal R(\mathcal S_k)$.
		The semantic accuracy of the recovered data
		\begin{equation}
			u_k(\mathcal D_k, \mathcal S_k)=\frac{\sum_{i=1}^{|\mathcal R(\mathcal S_k)|} \min\{\sigma(\mathcal R(\mathcal S_k),s_{ki}'),\sigma(\mathcal D_k,s_{ki}')\}}
			{\sum_{i=1}^{|\mathcal R(\mathcal S_k)|} \sigma(\mathcal R(\mathcal S_k),s_{ki}')},
		\end{equation}
		where $|\mathcal R(\mathcal S_k)|$ is the number of bits in $\mathcal R(\mathcal S_k)$, $s_{ki}'$ denotes the $i$-th word in text or frame in video in $\mathcal R(\mathcal S_k)$, and $\sigma(\mathcal R(\mathcal S_k),s_{ki}')$ is the number of occurrences of $s_{ki}'$ in $\mathcal R(\mathcal S_k)$.
		
		\subsection{RSMA Model}
		In RSMA,  the message intended for each user can be split into two parts, i.e., common part and private part  \cite{clerckx2019rate}.
		The common parts from all users are collected and combined into a common message.
		Through sharing the same codebook for all users, the common message is encoded into the common message $s_0$, which all users need to decode. 
		The private part of each user $k$ is encoded into the private stream $s_k$, which is intended for the specific user $k$.
		As a result, the transmitted signal $\boldsymbol x$ of the BS can be written as:
		\begin{equation}\label{miso1eq1}
			\boldsymbol x=\sqrt{p_0}\boldsymbol w_0 s_0  + \sum_{k=1}^K \sqrt{p_k} \boldsymbol  w_k s_k,
		\end{equation}
		where $\boldsymbol w_0$ is the transmit beamforming of the common message $s_0$ , $\boldsymbol w_k$ is the transmit beamforming  of the private message $s_k$ intended for user $k$, $p_0$ is the transmit power of the common message $s_0$,  and $p_k$ is the transmit power of the  private message $s_k$.
		
		For user $k$, the received message can be represented by:
		\begin{equation}\label{miso1eq2}
			\boldsymbol{h_k}^H \boldsymbol x + n_k=\boldsymbol h_k^H \sqrt{p_0}\boldsymbol w_0 s_0  +\sum_{j=1}^K \sqrt{p_k} \boldsymbol h_k^H  \boldsymbol w_k s_j+n_k,
		\end{equation}
		where $\boldsymbol h_k$ stands for the channel  between user $k$ and the BS. 
		To decode the common message $s_0$, the rate of user $k$ can be given by:
		\begin{equation}\label{miso1eq3}
			c_k =
			B \log_2 \left( 1+ \frac{p_0|\boldsymbol h_k^H \boldsymbol w_0|^2}{ \sum_{j=1}^Kp_j|\boldsymbol h_k^H \boldsymbol w_j|^2+\sigma^2}
			\right).
		\end{equation}
		where $B$ is the bandwidth of the BS.
		Note that all users need to decode the same common message. To ensure that all users can successfully decode the common message, the rate of the common message can be set as \cite{mao2018rate}  
		\begin{align}\label{sys1eq5}
			c_0=\min_{k\in\mathcal K}c_k.
		\end{align}
		
		In our considered SWC with rate splitting, the common knowledge is shared by all users. Thus, the common knowledge required for semantic communication can be encoded in the common message. 
		Besides, the common message also includes the parts that are allocated for different users, i.e., the rate in the common message allocated to user $k$ is denoted by $a_k$. As a result, the rate constraint for the common message can be given by
		\begin{equation}\label{miso1eq6}
			a_0+\sum_{k=1}^Ka_k\leq c_k, \quad \forall k\in\mathcal K,
		\end{equation}
		where $a_0$ is the rate allocated to updated common knowledge that all users need to receive. 
		In SWC, $a_0$ represents the rate of  transmitting the information of  updated directional probability graph.

		For each user, the common message is decoded first, and then the common message can be subtracted for decoding the private message. As a result, the rate for decoding the private message for user $k$ can be calculated as
		\begin{align}\label{miso1eq7}
			r_k &=
			B \log_2 \left( 1+ \frac{p_k|\boldsymbol h_k^H \boldsymbol w_k|^2}{ \sum_{j=1,j\neq k}^K p_j|\boldsymbol h_k^H \boldsymbol w_j|^2+\sigma^2}
			\right).
		\end{align}

		\subsection{Transmission and Computation Model}
		{\color{myc1}For each user $k$, the computation time for extracting semantic information from data $\mathcal D_k$ is
		\begin{equation}
			t_{1k}=\frac{ y_{1k}(\mathcal D_k, \mathcal S_k)}{f_k},
		\end{equation}
		where $y_{1k}(\mathcal D_k, \mathcal S_k)$ is the required amount of  CPU cycles for calculating $\mathcal S_k$ out of $\mathcal D_k$, and $f_k$ is the computing capacity of user $k$.
		The local computation energy can be given by:
		\begin{equation}
			E_{1k}=\kappa y_{1k}(\mathcal D_k, \mathcal S_k) f_k^2,
		\end{equation}
		where $\kappa$ is a constant coefficient to measure the effective switched capacitance.}
		
		With private  rate \eqref{miso1eq7} and allocated common rate $a_k$, the downlink transmission time for transmitting $\mathcal S_k$ is given by
		\begin{equation}\label{sys1_tk21}
			t_{2k1}=\frac{Z(\mathcal S_k)}{r_k+a_k},
		\end{equation}
		where $Z(\mathcal S_k)$ is the data size of set $\mathcal S_k$. 
		To transmit the renewed information about the knowledge base, i.e.,  updated information of directional probability graph, the transmission time of all users can be formulated as
		\begin{equation}\label{sys1_t022}
			t_{0}=\frac{K_0}{a_0},
		\end{equation} 
		where $K_0$ is the size of updated information of directional probability graph. 
		Combining \eqref{sys1_tk21}  and \eqref{sys1_t022}, the downlink transmission time for user $k$ is
		\begin{equation}\label{sys1_tk2}
			t_{2k}=\max\{t_{2k1}, t_{0}\}.
		\end{equation} 
		
		The  transmission energy for sending $\mathcal S_k$ is
		\begin{equation}
			E_{2k}=t_{2k1} p_k,
		\end{equation}
		and the transmission energy for broadcasting updated information of directional probability graph is
		\begin{equation}
			E_{20}=t_{0} p_0.
		\end{equation}
		
		At user $k$, to recover the original data, the user needs to compute the semantic information $\mathcal S_k$.
		The computation time of user $k$
		\begin{equation}
			t_{3k}=  \frac{y_{2k}(\mathcal D_k, \mathcal S_k)}{g_k},
		\end{equation}
		where $y_{2k}(\mathcal D_k, \mathcal S_k)$ is the number of computation cycles of recovering $\mathcal D_k$ from $\mathcal S_k$ and $g_k$ is the computation capacity at user $k$.
		The total complete time for user $k$ includes computation time at the BS, downlink transmission time, and computation time at user $k$, as shown in Fig. \ref{sys2}.
		The overall complete  time  of user $k$ including both computation and computation is 
		\begin{align}
			t_{k}&= t_{1k}+t_{2k}+t_{3k}
			\nonumber\\&= \frac{y_{1k}(\mathcal D_k, \mathcal S_k)}{f_k}+
			\max\left\{\frac{Z(\mathcal S_k)}{r_k+a_k}, \frac{K_0}{a_0}   \right\}
			+ \frac{y_{2k}(\mathcal D_k, \mathcal S_k)}{g_k}.
		\end{align}
		
		\begin{figure}
			\centering
			\includegraphics[width=3.5in]{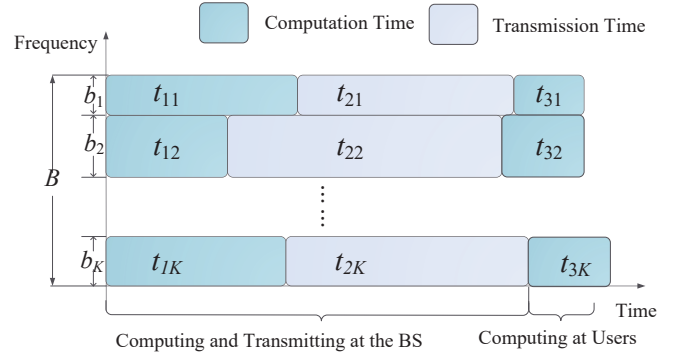}
			\vspace{-.5em}
			\caption{Illustration of the computation and communication time.}
			\vspace{-.5em}
			\label{sys2}
		\end{figure}
		The energy consumption at user $k$ is
		\begin{equation}
			E_{3k}= \kappa   y_{2k}(\mathcal D_k, \mathcal S_k)  g_k^2.
		\end{equation}
		
		With the above considered model, the total communication and computation energy consumption of the system is
		\begin{align}
			E=&\sum_{k=1}^K (E_{1k}+E_{2k}+E_{3k})+E_0
			\nonumber \\
			=&\sum_{k=1}^K \left( \kappa y_{1k}(\mathcal D_k, \mathcal S_k) f_k^2
			+  \frac{Z(\mathcal S_k)p_k}{r_k+a_k}
			+\kappa y_{2k}(\mathcal D_k, \mathcal S_k) g_k^2\right)
			\nonumber \\&
			+\frac{K_0p_0}{a_0}.
		\end{align}
		
		We aim to minimize the total energy consumption of the whole system with considering the completion time, transmit information accuracy,  computation capacity, rate allocation, and power allocation constraints. Mathematically, the formulated total energy minimization problem can be given by:
		\begin{subequations}\label{sys1min1}
			\begin{align}
				\mathop{\min}_{\mathcal S,  \boldsymbol f,  \boldsymbol g,  \boldsymbol p, \boldsymbol a, \boldsymbol w} \:&
				E,
				\tag{\theequation}  \\
				\textrm{s.t.} \:\:\:
				& \frac{y_{1k}(\mathcal D_k, \mathcal S_k)}{f_k}+
				\max\left\{\frac{Z(\mathcal S_k)}{r_k+a_k}, \frac{K_0}{a_0}   \right\}
				\nonumber\\&
				+ \frac{y_{2k}(\mathcal D_k, \mathcal S_k)}{g_k}\leq T,\quad \forall k\in \mathcal K,\\
				&  u_k(\mathcal D_k, \mathcal S_k)\geq A_k, \quad \forall k \in \mathcal K,\\
				& \mathcal S_k \subseteq  \mathcal G(\mathcal D_k) \quad \forall k \in \mathcal K,\\
				&a_0+	\sum_{k=1}^Ka_k\leq c_k, \quad \forall k\in\mathcal K, \\
				& \sum_{k=1}^K f_k\leq F^{\max}\\
				& \sum_{k=0}^K p_0\leq P^{\max}\\
				& a_k, f_k,p_k\geq 0,\quad \forall k,\\
				& \|\boldsymbol{w}_k\|=1,\quad \forall k \in \mathcal K\cup\{0\},\\
				&  0\leq g_k\leq g_k^{\max}, \quad \forall k \in \mathcal K,
			\end{align}
		\end{subequations}
		where
		$\mathcal S =\{\mathcal S_1, \cdots, \mathcal S_K\}$,
		$\boldsymbol f=[f_0, f_1, \cdots, f_K]^T$,
		$\boldsymbol g=[g_1, \cdots, g_K]^T$,
		$\boldsymbol p=[p_1, \cdots, p_K]^T$,
		$\boldsymbol a=[a_0, \cdots, p_K]^T$,
		$\boldsymbol w=[\boldsymbol w_0; \boldsymbol w_1;\cdots; \boldsymbol w_K]$,
		$T$ is the maximum communication delay of the system,
		$A_k$ is the minimum semantic accuracy for user $k$,
		$F^{\max}$ is the maximum computation capacity at the BS,
		$P^{\max}$ is the  transmission power of the BS,
		and 	$g_k^{\max}$ is the maximum local computation capacity of user $k$. 
		Since both objective function and constraints (\ref{sys1min1}a)-(\ref{sys1min1}c) are nonconvex, it is generally hard to solve this problem. 
		To solve this problem, we propose an iterative algorithm using the alternating method and successive convex approximation (SCA) approach.
		
		\section{Algorithm Design}
		In this section, an alternating algorithm is proposed to iteratively solve problem \eqref{sys1min1} through optimizing three subproblems, i.e., semantic information extraction subproblem,  computation capacity subproblem, joint power control, rate allocation and beamforming design subproblem. 
		\subsection{Semantic Information Extraction}
		With given computation capacity, power control, rate allocation, and beamforming design, problem \eqref{sys1min1}  can be simplified as
		\begin{subequations}\label{alg2min1}
			\begin{align}
				\mathop{\min}_{\mathcal S} \:&
				\sum_{k=1}^K \left( \kappa y_{1k}(\mathcal D_k, \mathcal S_k) f_k^2
				+  \frac{Z(\mathcal S_k)p_k}{r_k+a_k}
				+\kappa y_{2k}(\mathcal D_k, \mathcal S_k) g_k^2\right)
				\nonumber \\&
				+\frac{K_0p_0}{a_0}
				\tag{\theequation}  \\
				\textrm{s.t.} \:\:\:
				& \frac{y_{1k}(\mathcal D_k, \mathcal S_k)}{f_k}+
				\max\left\{\frac{Z(\mathcal S_k)}{r_k+a_k}, \frac{K_0}{a_0}   \right\}	\nonumber\\&
				+ \frac{y_{2k}(\mathcal D_k, \mathcal S_k)}{g_k}\leq T,\quad \forall k\in \mathcal K,\\
				&  u_k(\mathcal D_k, \mathcal S_k)\geq A_k, \quad \forall k \in \mathcal K,\\
				& \mathcal S_k \subseteq  \mathcal G(\mathcal D_k) \quad \forall k \in \mathcal K.
			\end{align}
		\end{subequations}
		
		Problem \eqref{alg2min1} is hard to solve because of  two general difficulties. 
		The first difficulty lies in the discrete value space of  variable $\mathcal S_k $, which leads to the discrete optimization problem and the complexity to find the optimal solution is usually extremely too high. 
		The second difficulty is the implicit expressions of accuracy function $u_k(\mathcal D_k, \mathcal S_k)$ and computation functions $f_{1k}(\mathcal D_k, \mathcal S_k)$ and $f_{2k}(\mathcal D_k, \mathcal S_k)$.
		\begin{figure}[t]
			\centering
			\includegraphics[width=3.8in]{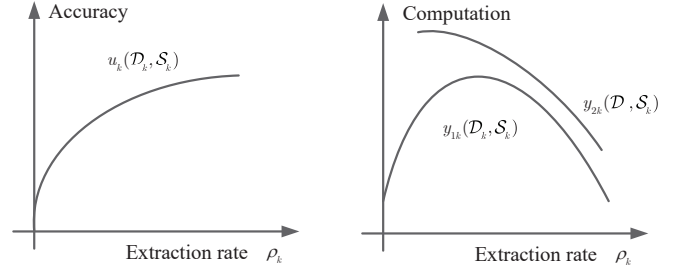}
			\vspace{-.5em}
			\caption{Illustration of the accuracy and computation functions versus the extraction rate.}
			\vspace{-.5em}
			\label{alg2fig1}
		\end{figure}
		
		To handle the first difficulty, we introduce the new variable, extraction rate $\rho_k$, which is defined as
		\begin{equation}\label{alg2min1eq1}
			\rho_k=\frac{Z(\mathcal S_k)}{Z(  \mathcal G(\mathcal D_k))}.
		\end{equation}
		Absolutely, the value of $\rho_k$ lies in (0,1]. 
		In the following, we use variable to replace  $\mathcal S_k$ for the purpose of obtaining the insights about extraction rate. 
		
		To handle the second difficulty, we first analyze the trend of accuracy and computation functions. 
		{\color{myc1}For the accuracy function, the accuracy always increases with the extraction rate since more information can be used to recover the original data, as shown in Fig. \ref{alg2fig1}. As a result, the minimum accuracy constraint (\ref{alg2min1}c) can be equivalent to 
		\begin{equation}\label{alg2min1eq2}
			\rho_k\geq \Gamma_k,
		\end{equation}
		where $\Gamma_k$ is the minimum extraction rate satisfying $ u_k(\mathcal D_k, \Gamma_k)= A_k$.
		For computation function $y_{1k}(\mathcal D_k,  \rho_k)$ is the number of required CPU cycles for computing the information with extraction rate $\rho_k$ out of $\mathcal D_k$,   $y_{1k}(\mathcal D_k, \mathcal S_k)$ includes two parts.  The first part is computing the directional probability graph, which can be modeled as a function only related to the size of $\mathcal D_k$, i.e., $y_{3k}(\mathcal D_k)$. 
		The second part is selecting the information with extraction rate $\rho_k$ out from the  directional probability graph. Since $\rho_k=0$ or $\rho_k=1$, the selection scheme is straightforward, which leads to the lowest computation cycles. Hence, the computation of the second part first increases and then decreases with the extraction rate $\rho_k$.
		As an example,  computation function $y_{1k}(\mathcal D_k, \rho_k)$ can be expressed as 
		\begin{equation}\label{alg2min1eq3}
			y_{1k}(\mathcal D_k, \rho_k)=y_{3k}(\mathcal D_k)+C_{k1}(\rho_k-C_{k2})^{C_{k3}},
		\end{equation}
		where $C_{k1}>0$, $C_{k2}\in(0,1)$, and $C_{k3}>0$ are constant parameters and theses parameters can be obtained through simulations.
		For computation function $y_{2k}(\mathcal D_k, \rho_k)$, the number of computation cycles decreases with $\rho_k$ since more semantic information can be helpful in recovering the original information. As an example, the computation function $y_{2k}(\mathcal D_k, \rho_k)$ can be expressed as
		\begin{equation}\label{alg2min1eq5}
			y_{2k}(\mathcal D_k, \rho_k)=C_{k4}\rho_k^{-C_{k5}},
		\end{equation}
		where $C_{k4}>0$ and $C_{k5}>0$ are constant parameters  through simulations.}
		
		With the above variable substitution \eqref{alg2min1eq1} and expressions \eqref{alg2min1eq2}-\eqref{alg2min1eq5}, problem \eqref{alg2min1} can be reformulated as:
		\begin{subequations}\label{alg2min1_2}
			\begin{align}
				\mathop{\min}_{\pmb \rho} \:&
				\sum_{k=1}^K \Big( \kappa  f_k^2(y_{3k}(\mathcal D_k)+C_{k1}(\rho_k-C_{k2})^{C_{k3}})
					\nonumber\\&+  \frac{Z(\mathcal G(\mathcal D_k))p_k\rho_k}{r_k+a_k}
				+\kappa C_{k4}\rho_k^{-C_{k5}}g_k^2\Big)
				\nonumber \\&
				+\frac{K_0p_0}{a_0}
				\tag{\theequation}  \\
				\textrm{s.t.} \:\:\:
				& \frac{y_{3k}(\mathcal D_k)+C_{k1}(\rho_k-C_{k2})^{C_{k3}}}{f_k}
				\nonumber\\&+
				\max\left\{\frac{Z(\mathcal G(\mathcal D_k))\rho_k}{r_k+a_k}, \frac{K_0}{a_0}   \right\}
				+ \frac{C_{k4}\rho_k^{-C_{k5}}}{g_k}\leq T,	\nonumber\\&\quad \forall k\in \mathcal K,\\
				& \Gamma_k\leq\rho_k\leq 1, \quad \forall k \in \mathcal K,
			\end{align}
		\end{subequations}
		where $\boldsymbol \rho=[\rho_1, \cdots, \rho_K]^T$.
		Since both objective function and feasible set are convex, problem \eqref{alg2min1_2} is a convex problem. Thus, we can apply the dual method to obtain the Karush-Kuhn-Tucker (KKT) point. To calculate the  solution of problem \eqref{alg2min1_2}, we can obtain the following theorem. 
		\begin{theorem}
			The optimal solution of problem \eqref{alg2min1_2} is 
			\begin{eqnarray}\label{alg2min1_2Th1eq1}
				\rho_k^* = \left\{ \begin{array}{ll}
					\rho_{k1}^*(\lambda_{1k1}) & \textrm{if $\rho_{k1}^*(\lambda_{11})\geq \frac{K_0(a_k+r_k)}{a_0Z(\mathcal G(\mathcal D_k))}$}\\
					\rho_{k2}^*(\lambda_{1k2})  & \textrm{if $\rho_{k2}^*(\lambda_{12})< \frac{K_0(a_k+r_k)}{a_0Z(\mathcal G(\mathcal D_k))}$}
				\end{array} \right.,
			\end{eqnarray}
			where $\rho_{k1}^*(\lambda_{1k})$ and  $\rho_{k2}^*(\lambda_{1k})$ are respectively the solutions to $\frac{\partial \mathcal L_1(\boldsymbol \rho, \lambda_{1k})}{\partial \rho_k}=0$ in \eqref{alg2min1_2eq2} and \eqref{alg2min1_2eq3},  $\lambda_{1k1}$ and $\lambda_{1k2}$ respectively satisfy
			\begin{align}\label{alg2min1_2Th1eq1cons1}
				&\frac{y_{3k}(\mathcal D_k)+C_{k1}(\rho_{k1}^*(\lambda_{1k1})|_{\Gamma_k}^1-C_{k2})^{C_{k3}}}{f_k}+
				\frac{Z(\mathcal G(\mathcal D_k)) }{r_k+a_k} 
				\nonumber\\&+ \frac{C_{k4}(\rho_{k1}^*(\lambda_{1k1})|_{\Gamma_k}^1)^{-C_{k5}}}{g_k} = T,
			\end{align}
			\begin{align}\label{alg2min1_2Th1eq1cons2}
				&\frac{y_{3k}(\mathcal D_k)+C_{k1}(\rho_{k2}^*(\lambda_{1k2})|_{\Gamma_k}^1-C_{k2})^{C_{k3}}}{f_k}+
				\frac{K_0}{a_0}
					\nonumber\\&+ \frac{C_{k4}(\rho_{k2}^*(\lambda_{1k2})|_{\Gamma_k}^1)^{-C_{k5}}}{g_k} = T,
			\end{align}
			with  $a|_b^c=\min\{\max\{a,b\},c\}$.
		\end{theorem}
		\begin{proof}
			Denoting $\lambda_{11}, \cdots, \lambda_{1K}>0$  as the Lagrange multiplier variables associated with constraint (\ref{alg2min1_2}a), we obtain the Lagrange function of problem \eqref{alg2min1_2} as
			\begin{eqnarray}\label{alg2min1_2eq1}
				\begin{aligned}
					&\mathcal L_1(\boldsymbol \rho, \boldsymbol \lambda_1)
					=	\sum_{k=1}^K \Big( \kappa  f_k^2(y_{3k}(\mathcal D_k)
					+C_{k1}(\rho_k-C_{k2})^{C_{k3}})
					\\&+  \frac{Z(\mathcal G(\mathcal D_k))p_k\rho_k}{r_k+a_k}
					+\kappa C_{k4}\rho_k^{-C_{k5}}g_k^2\Big)
					+\frac{K_0p_0}{a_0}	\\&+	\sum_{k=1}^K\lambda_{1k}\Big(
					\frac{y_{3k}(\mathcal D_k)+C_{k1}(\rho_k-C_{k2})^{C_{k3}}}{f_k}
					\\&+
					\max\left\{\frac{Z(\mathcal G(\mathcal D_k))\rho_k}{r_k+a_k}, \frac{K_0}{a_0}   \right\}
					\\
					&+ \frac{C_{k4}\rho_k^{-C_{k5}}}{g_k}-T
					\Big),
				\end{aligned}
			\end{eqnarray}
			where $\boldsymbol \lambda_1=[\lambda_{11},\cdots,\lambda_{1K}]^T$.
			The first derivative of \eqref{alg2min1_2eq1} becomes
			\begin{eqnarray}\label{alg2min1_2eq2}
				\begin{aligned}
					&\frac{\partial \mathcal L_1(\boldsymbol \rho, \boldsymbol \lambda_1)}{\partial \rho_k}
					= \kappa  f_k^2 C_{k1}C_{k3}(\rho_k-C_{k2})^{C_{k3}-1}
					\\&+  \frac{Z(\mathcal G(\mathcal D_k))p_k}{r_k+a_k}
					-\kappa C_{k4}C_{k5}\rho_k^{-C_{k5}-1}g_k^2
					\\&
					+\lambda_{1k}\Big(
					\frac{ C_{k1}C_{k3}(\rho_k-C_{k2})^{C_{k3-1}}}{f_k}+
					\frac{Z(\mathcal G(\mathcal D_k)) }{r_k+a_k} 
		\\&	-\frac{C_{k4}C_{k5}\rho_k^{-C_{k5}-1}}{g_k}
					\Big)
				\end{aligned}
			\end{eqnarray}
			for 
			\begin{equation}\label{alg2min1_2eq2_2}
				\rho_k \geq \frac{K_0(a_k+r_k)}{a_0Z(\mathcal G(\mathcal D_k))},
			\end{equation}
			and 
			\begin{eqnarray}\label{alg2min1_2eq3}
				\begin{aligned}
					&\frac{\partial \mathcal L_1(\boldsymbol \rho, \boldsymbol\lambda_1)}{\partial \rho_k}
					= \kappa  f_k^2 C_{k1}C_{k3}(\rho_k-C_{k2})^{C_{k3}-1}
					\\&+  \frac{Z(\mathcal G(\mathcal D_k))p_k}{r_k+a_k}
					-\kappa C_{k4}C_{k5}\rho_k^{-C_{k5}-1}g_k^2
					\\&
					+\lambda_{1k}\Big(
					\frac{ C_{k1}C_{k3}(\rho_k-C_{k2})^{C_{k3-1}}}{f_k}
					-\frac{C_{k4}C_{k5}\rho_k^{-C_{k5}-1}}{g_k}
					\Big)
				\end{aligned}
			\end{eqnarray}
			for 
			\begin{equation}\label{alg2min1_2eq3_2}
				\rho_k <\frac{K_0(a_k+r_k)}{a_0Z(\mathcal G(\mathcal D_k))},
			\end{equation}
			Denote the solution of $\frac{\partial \mathcal L_1(\boldsymbol \rho, \boldsymbol \lambda_1)}{\partial \rho_k}=0$ to equations \eqref{alg2min1_2eq2} and \eqref{alg2min1_2eq3} by $\rho_{k1}^*(\lambda_1)$ and  $\rho_{k2}^*(\lambda_{1k})$, respectively. 
			Note that the left hand sides of \eqref{alg2min1_2eq2} and \eqref{alg2min1_2eq3} are monotonically increasing with respect to $\rho_k$,  solutions $\rho_{k1}^*(\lambda_{1k})$ and  $\rho_{k2}^*(\lambda_{1k})$ can be obtained via the bisection method. 
			Considering constraints (\ref{alg2min1_2}b), \eqref{alg2min1_2eq2_2}, and \eqref{alg2min1_2eq3_2}, the Lagrange multiplier should meet the KKT condition, i.e., the optimal solution of problem can be presented in  \eqref{alg2min1_2Th1eq1}.
		\end{proof}

		\subsection{Optimal Computation Capacity}
		With given semantic information extraction,  power control, rate allocation, and beamforming design, problem \eqref{sys1min1}  can be simplified as
		\begin{subequations}\label{alg2min2}
			\begin{align}
				\mathop{\min}_{\boldsymbol f,  \boldsymbol g} \:&
				\sum_{k=1}^K \left( \kappa y_{1k}(\mathcal D_k, \mathcal S_k) f_k^2
				+  \frac{Z(\mathcal S_k)p_k}{r_k+a_k}
				+\kappa y_{2k}(\mathcal D_k, \mathcal S_k) g_k^2\right)
				\nonumber \\&
				+\frac{K_0p_0}{a_0}
				\tag{\theequation}  \\
				\textrm{s.t.} \:\:\:
				& \frac{y_{1k}(\mathcal D_k, \mathcal S_k)}{f_k}+
				\max\left\{\frac{Z(\mathcal S_k)}{r_k+a_k}, \frac{K_0}{a_0}   \right\}	\nonumber\\&
				+ \frac{y_{2k}(\mathcal D_k, \mathcal S_k)}{g_k}\leq T,\quad \forall k\in \mathcal K,\\
				& \sum_{k=1}^K f_k\leq F^{\max}\\
				&f_k\geq 0,\quad \forall k,\\
				&  0\leq g_k\leq g_k^{\max}, \quad \forall k \in \mathcal K.
			\end{align}
		\end{subequations}
		The Language function of  problem \eqref{alg2min2} can be given by
		\begin{eqnarray}\label{alg2min2_2eq1_2}
			\begin{aligned}
			&	\mathcal L_2(\boldsymbol f, \boldsymbol g, \boldsymbol \lambda_2,\lambda_3)
				=	\sum_{k=1}^K \Big( \kappa y_{1k}(\mathcal D_k, \mathcal S_k) f_k^2
				+  \frac{Z(\mathcal S_k)p_k}{r_k+a_k}
					\\&+\kappa y_{2k}(\mathcal D_k, \mathcal S_k) g_k^2\Big)
				+\frac{K_0p_0}{a_0}
				\\&
				+\sum_{k=1}^K\lambda_{2k}\Big(\frac{y_{1k}(\mathcal D_k, \mathcal S_k)}{f_k}+
				\max\left\{\frac{Z(\mathcal S_k)}{r_k+a_k}, \frac{K_0}{a_0}   \right\}	\\&
				+ \frac{y_{2k}(\mathcal D_k, \mathcal S_k)}{g_k}- T\Big)
				+\lambda_3\left(\sum_{k=1}^K f_k- F^{\max}\right),
			\end{aligned}
		\end{eqnarray}
		where  $\boldsymbol \lambda_2=[\lambda_{21},\cdots,\lambda_{2K}]^T$ is the Language multiplier associated with constraint (\ref{alg2min2}a) and $\lambda_3>0$  is the Language multiplier associated with constraint (\ref{alg2min2}b). 
		The first derivative of \eqref{alg2min2_2eq1_2} becomes
		\begin{eqnarray}\label{alg2min2_2eq2}
			\begin{aligned}
				&\frac{\partial \mathcal L_2(\boldsymbol f, \boldsymbol g, \boldsymbol\lambda_2,\lambda_3)}{\partial f_k}
			\\	= & 2 \kappa y_{1k}(\mathcal D_k, \mathcal S_k) g_k- \frac{\lambda_{2k}y_{1k}(\mathcal D_k, \mathcal S_k)}{g_k^2}+\lambda_3
			\end{aligned}
		\end{eqnarray}
		\begin{eqnarray}\label{alg2min2_2eq3}
			\begin{aligned}
				&\frac{\partial \mathcal L_2(\boldsymbol f, \boldsymbol g,\boldsymbol\lambda_2,\lambda_3)}{\partial g_k}\\
				=  &2 \kappa y_{2k}(\mathcal D_k, \mathcal S_k) f_k- \frac{\lambda_{2k}y_{2k}(\mathcal D_k, \mathcal S_k)}{f_k^2}
			\end{aligned}
		\end{eqnarray}
		Setting $\frac{\partial \mathcal L_2(\boldsymbol f, \boldsymbol g,\boldsymbol\lambda_2,\lambda_3)}{\partial f_k}=0$ and $\frac{\partial \mathcal L_2(\boldsymbol f, \boldsymbol g,\boldsymbol\lambda_2,\lambda_3)}{\partial g_k}=0$ yields
		\begin{eqnarray}\label{alg2min2_2eq5}
			2 \kappa y_{1k}(\mathcal D_k, \mathcal S_k) f_k^3+\lambda_3 f_k^2- \lambda_{2k}y_{1k}(\mathcal D_k, \mathcal S_k)=0,
		\end{eqnarray}
		\begin{eqnarray}
			g_k= \left(\frac{\lambda_{2k}y_{2k}(\mathcal D_k, \mathcal S_k)}{  2 \kappa y_{2k}(\mathcal D_k, \mathcal S_k)}\right)^{\frac 13}.
		\end{eqnarray}
		The value of $f_k$ can be obtained via solving the cubic function in \eqref{alg2min2_2eq5}. 
		Having obtained the value of computation capacity $f_k$ and $g_k$, the value of Language multiplier can be updated via the gradient method. 
		In the $t$-th iteration, the value of $\lambda_{2k}$ and $\lambda_3$ are updated by
		\begin{align}
			\lambda_{2k}(t)=&\Bigg[\lambda_{2k}(t-1)-\upsilon(t)\Big(\frac{y_{1k}(\mathcal D_k, \mathcal S_k)}{f_k}+\nonumber\\&
			\max\left\{\frac{Z(\mathcal S_k)}{r_k+a_k}, \frac{K_0}{a_0}   \right\}
			+ \frac{y_{2k}(\mathcal D_k, \mathcal S_k)}{g_k}- T\Big)\Bigg]^+,
		\end{align}
	and
		\begin{equation}
			\lambda_3(t)=\left[\lambda_3(t-1)-\upsilon(t)\left(\sum_{k=1}^K f_k- F^{\max}\right)\right]^+,
		\end{equation}
		where $[a]^+=\max{a,0}$ and $\upsilon(t)>0$ is the dynamic step size. 
		Through iteratively updating $(f_k,g_k)$ and $(\lambda_{2k},\lambda_3)$, the overall procedure yields the global optimal solution of problem \eqref{alg2min2}.
		
		\subsection{Joint Power Control, Rate Allocation, and Beamforming Design}
		With given semantic information extraction and computation capacity, problem \eqref{sys1min1}  can be simplified as
		\begin{subequations}\label{miso1min0}
			\begin{align}
				\mathop{\min}_{ \boldsymbol p, \boldsymbol a, \boldsymbol w} \:&
				\sum_{k=1}^K \left( \kappa y_{1k}(\mathcal D_k, \mathcal S_k) f_k^2
				+  \frac{Z(\mathcal S_k)p_k}{r_k+a_k}
				+\kappa y_{2k}(\mathcal D_k, \mathcal S_k) g_k^2\right)
				\nonumber \\&
				+\frac{K_0p_0}{a_0},
				\tag{\theequation}  \\
				\textrm{s.t.} \:\:\:
				& \frac{y_{1k}(\mathcal D_k, \mathcal S_k)}{f_k}+
				\max\left\{\frac{Z(\mathcal S_k)}{r_k+a_k}, \frac{K_0}{a_0}   \right\}\nonumber\\&
				+ \frac{y_{2k}(\mathcal D_k, \mathcal S_k)}{g_k}\leq T,\quad \forall k\in \mathcal K,\\
				&a_0+	\sum_{k=1}^Ka_k\leq c_k, \quad \forall k\in\mathcal K, \\
				& \sum_{k=0}^K p_0\leq P^{\max}\\
				& a_0, a_k, p_k\geq 0,\quad \forall k,\\
				& \|\boldsymbol{w}_k\|=1,\quad \forall k \in \mathcal K\cup\{0\},
			\end{align}
		\end{subequations}
		
		Problem \eqref{miso1min0} is nonconvex owing  to the nonconvex objective function and constraints (\ref{miso1min0}a),  (\ref{miso1min0}b) and  (\ref{miso1min0}e).
		{\color{myc1}To handle the nonconvexity of the objective function, we introduce new variable $r_k$ and use variable $p_k^2$ to replace power $p_k$.} Thus, problem \eqref{miso1min0}  can be equivalently transformed to
		\begin{subequations}\label{miso1min0_2}
			\begin{align}
				\mathop{\min}_{ \boldsymbol p, \boldsymbol a, \boldsymbol r, \boldsymbol w} \:&
				\sum_{k=1}^K \left( \kappa y_{1k}(\mathcal D_k, \mathcal S_k) f_k^2
				+  \frac{Z(\mathcal S_k)p_k^2}{r_k+a_k}
				+\kappa y_{2k}(\mathcal D_k, \mathcal S_k) g_k^2\right)
				\nonumber \\&
				+\frac{K_0p_0^2}{a_0},
				\tag{\theequation}  \\
				\textrm{s.t.} \:\:\:
				& \frac{y_{1k}(\mathcal D_k, \mathcal S_k)}{f_k}+
				\max\left\{\frac{Z(\mathcal S_k)}{r_k+a_k}, \frac{K_0}{a_0}   \right\}	\nonumber \\&
				+ \frac{y_{2k}(\mathcal D_k, \mathcal S_k)}{g_k}\leq T,\quad \forall k\in \mathcal K,\\
				&a_0+	\sum_{k=1}^Ka_k\leq B \log_2 \left( 1+ \frac{p_0^2|\boldsymbol h_k^H \boldsymbol w_0|^2}{ \sum_{j=1}^Kp_j^2|\boldsymbol h_k^H \boldsymbol w_j|^2+\sigma^2}
				\right), 	\nonumber \\&\quad \forall k\in\mathcal K, \\
				&r_k\leq B \log_2 \left( 1+ \frac{p_k^2|\boldsymbol h_k^H \boldsymbol w_k|^2}{ \sum_{j=1,j\neq k}^K p_j^2|\boldsymbol h_k^H \boldsymbol w_j|^2+\sigma^2}
				\right),	\nonumber \\& \quad \forall k\in\mathcal K, \\
				& a_0, a_k, p_k\geq 0,\quad \forall k,\\
				& \|\boldsymbol{w}_k\|\leq1,\quad \forall k \in \mathcal K\cup\{0\},
			\end{align}
		\end{subequations}
		where $\boldsymbol r=[r_0,r_1,\cdots, r_K]^T$, the objective function is convex,
		and  constraint (\ref{miso1min0_2}e) is replaced  by the inequality  without loss of generality. 
		In problem \eqref{miso1min0_2}, we only need to deal with the nonconvexity of constraints (\ref{miso1min0_2}b)  and (\ref{miso1min0_2}c) .
		Through	introducing slacking variables $\gamma_k$ and $\eta_k$, problem \eqref{miso1min0_2} can be reformulated as:
		\begin{subequations}\label{miso1min2}
			\begin{align}
				\mathop{\min}_{ \boldsymbol p, \boldsymbol a, \boldsymbol r, \boldsymbol w,\boldsymbol  \gamma, \boldsymbol  \eta} \:&
				\sum_{k=1}^K \Big( \kappa y_{1k}(\mathcal D_k, \mathcal S_k) f_k^2
				+  \frac{Z(\mathcal S_k)p_k^2}{r_k+a_k}
				\nonumber \\&+\kappa y_{2k}(\mathcal D_k, \mathcal S_k) g_k^2\Big)
				+\frac{K_0p_0^2}{a_0},
				\tag{\theequation}  \\
				\textrm{s.t.} \:\:\:
				& \frac{y_{1k}(\mathcal D_k, \mathcal S_k)}{f_k}+
				\max\left\{\frac{Z(\mathcal S_k)}{r_k+a_k}, \frac{K_0}{a_0}   \right\}	\nonumber \\&
				+ \frac{y_{2k}(\mathcal D_k, \mathcal S_k)}{g_k}\leq T,\quad \forall k\in \mathcal K,\\
				&a_0+	\sum_{k=1}^Ka_k\leq B \log_2 \left( 1+ \eta_k
				\right), \quad \forall k\in\mathcal K, \\
				&r_k\leq B \log_2 \left( 1+  \gamma_k
				\right), \quad \forall k\in\mathcal K, \\
				& a_0, a_k, p_k\geq 0,\quad \forall k,\\
				& \|\boldsymbol{w}_k\|\leq1,\quad \forall k \in \mathcal K\cup\{0\},\\
				&\frac{p_k^2|\boldsymbol h_k^H \boldsymbol w_k|^2}{ \sum_{j=1,j\neq k}^K p_j^2 |\boldsymbol h_k^H \boldsymbol w_j|^2+\sigma^2}\geq\gamma_k,\quad\forall k \in \mathcal K,\\
				&\frac{p_0^2|\boldsymbol h_k^H \boldsymbol w_0|^2}{ \sum_{j=1}^Kp_j^2 |\boldsymbol h_k^H \boldsymbol w_j|^2+\sigma^2}\geq \eta_k,\quad\forall k \in \mathcal K,
			\end{align}
		\end{subequations}
		where $\boldsymbol \gamma=[\gamma_1,\cdots,\gamma_K]^T$ and $\boldsymbol \eta=[\eta_1,\cdots,\eta_K]^T$.
		In problem \eqref{miso1min2}, the objective function is transformed into convex.
		Because of nonconex constraints (\ref{miso1min2}f) and (\ref{miso1min2}g), problem \eqref{miso1min2} is nonconvex. In the following, we utilize the SCA method to handle these two nonconvex constraints. 
		
		{\color{myc1}For constraint (\ref{miso1min2}f),  it can be equivalent to 
		\begin{equation}\label{miso1min2eq1}
			p_k^2|\boldsymbol h_k^H \boldsymbol w_k|^2\geq \gamma_k \alpha_k,
		\end{equation}
		\begin{equation}\label{miso1min2eq2}
			\sum_{j=1,j\neq k}^Kp_j^2|\boldsymbol h_k^H \boldsymbol w_j|^2+\sigma^2 \leq \alpha_k,
		\end{equation}
	where $\alpha_k$ is a nonnegative slack variable. 
	In \eqref{miso1min2eq1}, we can always choose the term  $\boldsymbol h_k^H \boldsymbol w_k$ as a real value through changing the phase of  beamforming $\boldsymbol w_k$.
Thus,  constraint \eqref{miso1min2eq1} can be rewritten as
	 \begin{equation}\label{miso1min2eq1_2}
		\mathcal R(\boldsymbol h_k^H \boldsymbol w_k)\geq \frac{\sqrt{\gamma_k \alpha_k}}{p_k},
\end{equation}
where the left hand side is convex now. 
Through using the first-order Taylor series to replace the right hand side of \eqref{miso1min2eq1_2}, constraint \eqref{miso1min2eq1_2} can be approximated by 
	 \begin{align}\label{miso1min2eq1_22}
&	\mathcal R(\boldsymbol h_k^H \boldsymbol w_k)\geq \frac{\sqrt{\gamma_k^{(n)} \alpha_k^{(n)}}}{p_k^{(n)}}+
	\frac{\sqrt{\gamma_k^{(n)}  }}{2p_k^{(n)} \sqrt{\alpha_k^{(n)}}}(\alpha_k-\alpha_k^{(n)})	\nonumber \\&
	+	\frac{\sqrt{\alpha_k^{(n)}  }}{2p_k^{(n)} \sqrt{\gamma_k^{(n)}}}(\gamma_k-\gamma_k^{(n)})
	-	\frac{\sqrt{\gamma_k^{(n)}\alpha_k^{(n)}  }}{(p_k^{(n)})^2 }(p_k-p_k^{(n)}),
\end{align}
where the superscript $(n)$ means the value of the variable in the $n$-th iteration. 
 Moreover, \eqref{miso1min2eq2} can be reformulated as
\begin{align}\label{miso1min2eq2_2}
&\sum_{j=1,j\neq k}^K\frac 1 4 (({p_j^2}+|\boldsymbol h_k^H \boldsymbol w_j|^2)^2-({p_j^2}-|\boldsymbol h_k^H \boldsymbol w_j|^2)^2)
	\nonumber \\&=\sum_{j=1,j\neq k}^K{p_j^2}|\boldsymbol h_k^H \boldsymbol w_j|^2+ {\sigma^2}  \leq  {\alpha_k}.
\end{align}
Through replacing the left hand side of \eqref{miso1min2eq2_2} with its first-order Taylor approximation, we can obtain 
	\begin{align}\label{miso1min2eq2_22}
	&\sum_{j=1,j\neq k}^K \frac 1 4
	\bigg[ ((p_j^{(n)})^2+	|\boldsymbol h_k^H \boldsymbol w_j^{(n)}|^2)^2
		+4((p_j^{(n)})^2	\nonumber \\
		&+	|\boldsymbol h_k^H \boldsymbol w_j^{(n)}|^2)p_j^{(n)}(p_j-p_j^{(n)}) 	-(p_j^2- |\boldsymbol h_k^H \boldsymbol w_j|^2)^2
	\nonumber \\
			&+4((p_j^{(n)})^2+	|\boldsymbol h_k^H \boldsymbol w_j^{(n)}|^2)(\mathcal R(\boldsymbol h_k^H \boldsymbol w_j^{(n)}\boldsymbol h_k^H \boldsymbol w_j)-
			|\boldsymbol h_k^H \boldsymbol w_j^{(n)}|^2)
	\bigg]	\nonumber \\
	&
	+{\sigma^2}  \leq  {\alpha_k}.
\end{align}
		
Similarly, we  can introduce slack variable $\beta_k$ and constraint  (\ref{miso1min2}g) can be rewritten as:
		\begin{align}\label{miso1min2eq1_5_1}
			&\frac 1 4((p_0^2+\	|\boldsymbol h_k^H \boldsymbol w_0|^2)^2-(p_0^2- 	|\boldsymbol h_k^H \boldsymbol w_0|^2)^2)	\nonumber \\
			=&p_0^2	|\boldsymbol h_k^H \boldsymbol w_0|^2\geq \beta_k \eta_k=\frac 1 4((\beta_k+\eta_k)^2-(\beta_k- \eta_k)^2),
		\end{align}
		\begin{equation}\label{miso1min2eq1_5_2}
			\sum_{j=1}^Kp_j^2|\boldsymbol h_k^H \boldsymbol w_j|^2+\sigma^2 \leq \beta_k.
		\end{equation}
	Note that we cannot make  $\boldsymbol h_k^H \boldsymbol w_0$ as real values for all $k$ through changing the phase of $\boldsymbol w_0$. To handle the nonconvexity of \eqref{miso1min2eq1_5_1}, we use   first-order Taylor approximation on both sides of \eqref{miso1min2eq1_5_1}, which is different from the method in \cite{yang2021optimization}.
	Considering the  first-order Taylor approximation on both sides, \eqref{miso1min2eq1_5_1} can be transformed to 
			\begin{align}\label{miso1min2eq1_5_11}
		&((p_0^{(n)})^2+	|\boldsymbol h_k^H \boldsymbol w_0^{(n)}|^2)^2
		+4((p_0^{(n)})^2	\nonumber \\
		+&	|\boldsymbol h_k^H \boldsymbol w_0^{(n)}|^2)p_0^{(n)}(p_0-p_0^{(n)}) 	-(p_0^2- |\boldsymbol h_k^H \boldsymbol w_0|^2)^2
			\nonumber \\
			+&4((p_0^{(n)})^2+	|\boldsymbol h_k^H \boldsymbol w_0^{(n)}|^2)(\mathcal R(\boldsymbol h_k^H \boldsymbol w_0^{(n)}\boldsymbol h_k^H \boldsymbol w_0)-
			|\boldsymbol h_k^H \boldsymbol w_0^{(n)}|^2)
			\nonumber \\
		\geq &(\beta_k+\eta_k)^2-(\beta_k^{(n)}- \eta_k^{(n)})(\beta_k- \eta_k)
	 +(\beta_k^{(n)}- \eta_k^{(n)})^2,
	\end{align}
	For constraint \eqref{miso1min2eq1_5_2}, we can use the similar method to handle the nonconvexity of \eqref{miso1min2eq2}. Thus, \eqref{miso1min2eq1_5_2} can be approximated by 
	\begin{align}\label{miso1min2eq1_5_21}
	&\sum_{j=1}^K \frac 1 4
	\bigg[ ((p_j^{(n)})^2+	|\boldsymbol h_k^H \boldsymbol w_j^{(n)}|^2)^2
		+4((p_j^{(n)})^2	\nonumber \\
		+&	|\boldsymbol h_k^H \boldsymbol w_j^{(n)}|^2)p_j^{(n)}(p_j-p_j^{(n)}) 	-(p_j^2- |\boldsymbol h_k^H \boldsymbol w_j|^2)^2
	\nonumber \\
			+&4((p_j^{(n)})^2+	|\boldsymbol h_k^H \boldsymbol w_j^{(n)}|^2)(\mathcal R(\boldsymbol h_k^H \boldsymbol w_j^{(n)}\boldsymbol h_k^H \boldsymbol w_j)-
			|\boldsymbol h_k^H \boldsymbol w_j^{(n)}|^2)
	\bigg]	\nonumber \\
	&
	+{\sigma^2}  \leq  {\beta_k}.
		\end{align}
		
		With the above approximations, we can approximate the nonconvex constraints (\ref{miso1min2}f) and (\ref{miso1min2}g) with the corresponding convex approximation terms. Thus, the original problem  (\ref{miso1min2}) can be approximated by the following convex one:
				\begin{subequations}\label{miso1min3}
			\begin{align}
				\mathop{\min}_{ \boldsymbol p, \boldsymbol a, \boldsymbol r, \boldsymbol w,\boldsymbol  \gamma, \boldsymbol  \alpha, \boldsymbol  \eta,\boldsymbol  \beta} \:&
				\sum_{k=1}^K \bigg( \kappa y_{1k}(\mathcal D_k, \mathcal S_k) f_k^2
				+  \frac{Z(\mathcal S_k)p_k^2}{r_k+a_k}
			\nonumber \\&	+\kappa y_{2k}(\mathcal D_k, \mathcal S_k) g_k^2\bigg)
				+\frac{K_0p_0^2}{a_0},
				\tag{\theequation}  \\
				\textrm{s.t.} \:\:\:
				& (\ref{miso1min2}a)-(\ref{miso1min2}e), \eqref{miso1min2eq1_22}, \eqref{miso1min2eq2_22}, \eqref{miso1min2eq1_5_11}, \eqref{miso1min2eq1_5_21},\\
				&
				\alpha_k\geq0, \beta_k\geq0,  \quad \forall k \in \mathcal K,
			\end{align}
		\end{subequations}
	where $\boldsymbol \alpha=[\alpha,\cdots,\alpha]^T$ and $\boldsymbol \beta=[\beta_1,\cdots,\beta_K]^T$.
	The convex problem \eqref{miso1min3} can be solved by the existing convex optimization toolbox. }
		
		\subsection{Algorithm Analysis}
		
		{\color{myc1}The overall joint communication and computation resource allocation for SWC with RSMA is presented in Algorithm 1. 
		According to Algorithm 1, the complexity of solving problem \eqref{sys1min1} lies in solving three subproblems at each iteration. 
		For the semantic information extraction subproblem, the optimal solution is calculated by \eqref{alg2min1_2Th1eq1} in Theorem 1 with complexity $\mathcal O(K \log_2(1/\epsilon_1))$, where $\mathcal O(\log_2(1/\epsilon_1))$ is the complexity of solving \eqref{alg2min1_2Th1eq1cons1} and \eqref{alg2min1_2Th1eq1cons2} with the bisection method of accuracy $\epsilon_1$. 
		For the computation capacity subproblem, the complexity is $\mathcal O(N_1K)$, where $N_1$  denotes the number of iterations of using the dual method for solving the  computation capacity subproblem.
		For the joint power control, rate allocation, and beamforming design subproblem, the complexity lies in solving the approximated convex problem  \eqref{miso1min3}. 
		The complexity of obtaining the solution of  problem  \eqref{miso1min3} is $\mathcal O(M_1^2M_2)$ \cite{lobo1998applications}, where $M_1=(N+7)K+N+2$ is the total number of variables and $M_2=13K+1$ is the total number of constraints. The total complexity of solving the joint power control, rate allocation, and beamforming design subproblem is $\mathcal O(N_2 N^2K^3)$, where $N_2$ is the number of iterations for the SCA method. 
		As a result, the total complexity of the proposed Algorithm 1 is $\mathcal O(N_3K\log_2(1/\epsilon_1)+N_1N_3K+N_2N_3 N^2K^3)$, where $N_3$ is the number of outer iterations of Algorithm 1.}

		\begin{algorithm}[t]
			\caption{Joint Communication and Computation Resource Allocation for SWC with RSMA}
			\begin{algorithmic}[1]
				\STATE Initialize $ \mathcal S^{(0)},  \boldsymbol f^{(0)},  \boldsymbol g^{(0)},  \boldsymbol p^{(0)}, \boldsymbol a^{(0)}, \boldsymbol w^{(0)}$. Set iteration number $n=1$.
				\REPEAT
				\STATE With given $\boldsymbol f^{(n-1)},  \boldsymbol g^{(n-1)},  \boldsymbol p^{(n-1)}, \boldsymbol a^{(n-1)}, \boldsymbol w^{(n-1)}$, solve the semantic information extraction subproblem and obtain the solution $\mathcal S^{(n)}$.
				\STATE With given $\mathcal S^{(n)},  \boldsymbol p^{(n-1)}, \boldsymbol a^{(n-1)}, \boldsymbol w^{(n-1)}$, solve the computation capacity subproblem and obtain the solution $\boldsymbol f^{(n)},  \boldsymbol g^{(n)}$.
				\STATE With given $\mathcal S^{(n)},  \boldsymbol f^{(n)},  \boldsymbol g^{(n)}$, solve the joint power control, rate allocation, and beamforming design subproblem, of which  the solution is $\boldsymbol p^{(n)}, \boldsymbol a^{(n)}, \boldsymbol w^{(n)}$.
				\STATE Set $n=n+1$.
				\UNTIL the objective value (\ref{sys1min1}) converges.
			\end{algorithmic}
		\end{algorithm}

		\section{Simulation Results}
		In the simulations, there are $K=5$ users in the considered area. 
		For the pathloss model between each user and the BS, we set $128.1+37.6\log_{10} d$ ($d$ is in km) \cite{8352643}
		and the standard deviation of shadow fading is $4$ dB \cite{yang2021optimization}.
	Furthermore, the total bandwidth of the system is $B=20$ MHz and the power spectral density of the noise power is  $-174$ dBm/Hz.
	Unless specified otherwise, we set maximum transmit power $P^{\max}=30$ dBm, the effective switched capacitance in local computation is $\kappa=10^{-28}$,   maximum local computation capacity $g_1^{\max}=\cdots=g_K^{\max}=2$ GHz.
		For the considered semantic information task, we consider the same parameters as in \cite{ammar2018construction}.
		The main system parameters are summarized in Table~I.
		
		\begin{table}[t]
			\centering
			\caption{Main System  Parameters} \label{tab:complexity}
			\begin{tabular}{ccc}
				\hline
				\hline
				Parameter &   Value \\ \hline
				Bandwidth of the BS $B$ & 20 MHz  \\
				Power spectral density of the noise power  & -174 dBm/Hz \\
				Maximum transmit power $P^{\max}$ & 30 dBm\\
				Effective switched capacitance $\kappa$& $10^{-28}$\\
				Number of users $K$& $5$\\
				\hline
				\hline
			\end{tabular}
		\end{table}
		
		The proposed joint communication and computation resource allocation for SWC with  RSMA is labeled as `RSMA'.
		To compare the results of the proposed scheme, we consider the conventional orthogonal multiple access, frequency division multiple access (FDMA) \cite{seong2006optimal}, which is labeled as `FDMA', the total energy minimization problem for NOMA \cite{yang2019efficient}, which is labeled as `NOMA'.
		{\color{myc1}To better show the performance of multiple antenna scheme, we consider the SDMA system as in \cite{mao2018rate}.}

		\begin{figure}[t]
			\centering
			\includegraphics[width=3.5in]{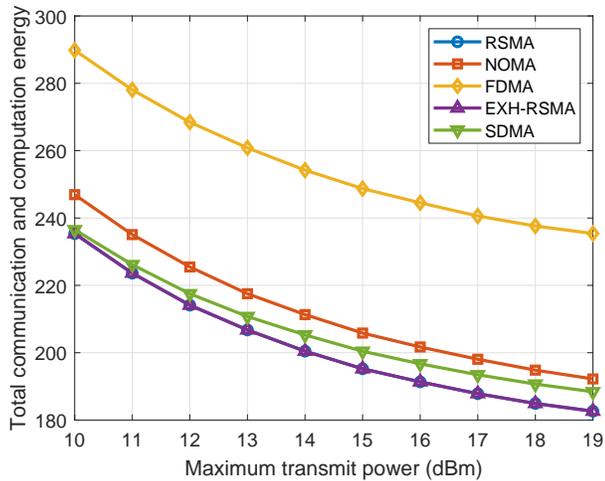}
			\vspace{-0.5em}
			\caption{Total communication and computation energy vs. maximum transmit power.} \label{fig6}
			\vspace{-0.5em}
		\end{figure}
		
		Fig.~\ref{fig6}  illustrates that the total communication and computation energy changes as the maximum transmit power of each user varies.
		According to this figure, the EXH-RSMA scheme stands for the exhaustive search method, which can yield a near globally optimal solution through running the proposed algorithm with  1000 initial solutions.
		It can be shown from this figure that the total energy decreases with the maximum transmit power of the BS. This is due to the fact that large transmit power can lead to low transmit time, which allows more time for computation and yields low total energy consumption. 
		It is observed that the proposed RSMA outperforms FDMA, NOMA, since RSMA can achieve higher spectral efficiency than FDMA and NOMA.
		{\color{myc1}Compared to SDMA, RSMA can still achieve better energy consumption, in particular the maximum transmit power is high.  The reason is that 
		SDMA is more likely to serve the users with higher channel gains, while the users with poor channel gains tend to have long transmit time and high computation power is needed for task computation, thus leading to higher total energy consumption than RSMA.}  It can be also found that the proposed RSMA achieves near performance as the EXH-RSMA, which indicates the effectiveness of the proposed scheme.

		\begin{figure}[t]
			\centering
			\includegraphics[width=3.5in]{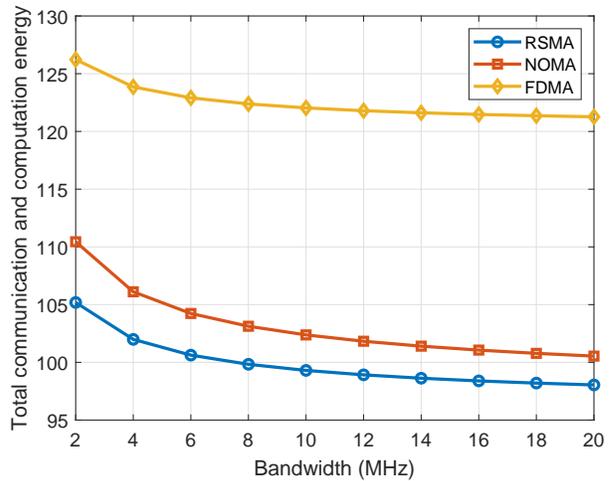}
			\vspace{-0.5em}
			\caption{Total communication and computation energy  vs. bandwidth of the system.} \label{fig7}
			\vspace{-0.5em}
		\end{figure}

		Fig. \ref{fig7} shows the total energy  versus bandwidth of the system.
		Based on this figure, the total communication and computation energy decreases as the bandwidth of the system increases for all schemes.
		This is because high bandwidth decreases the transmit time between users and the BS, which allows long  computation time and consequently reduces the local computation energy consumption.
		
		\begin{figure}[t]
			\centering
			\includegraphics[width=3.5in]{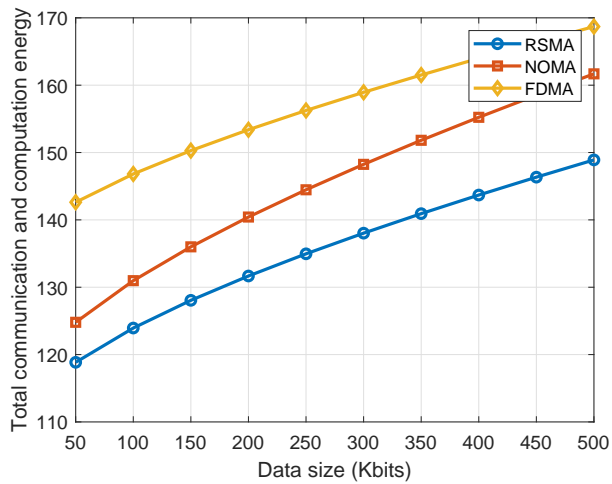}
			\vspace{-0.5em}
			\caption{Total communication and computation energy vs. transmit data size of each user.} \label{fig8}
			\vspace{-0.5em}
		\end{figure}

		 Fig. \ref{fig8} illustrates the trend of total communication and computation energy with the transmit data size of each user. 
		It is observed that the total energy increases as the data size for all schemes.
		This is due to the fact that more information needs to be transmitted, thus increasing the transmit and computation power. It can be found that the growing speed of total energy od the proposed RSMA is slower than that of NOMA and FDMA, which shows the robustness of the RSMA.

		\begin{figure}[t]
			\centering
			\includegraphics[width=3.5in]{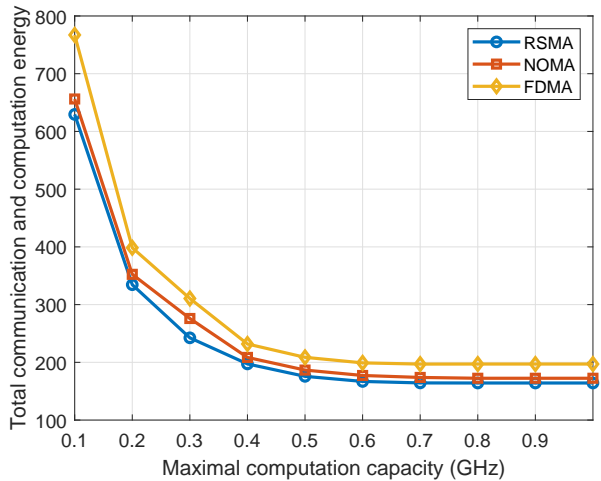}
			\vspace{-0.5em}
			\caption{Total communication and computation energy vs. maximum computation capacity of each user.} \label{fig9}
			\vspace{-0.5em}
		\end{figure}
		
		To show how the computation capacity affects the system performance, Fig.~\ref{fig9} presents the total communication and computation energy versus the maximum computation capacity of each user. According to this figure, the total energy first decreases rapidly and then the total energy tends to approach a fixed value.  The reason lies in that for small computation capacity region, the increase of maximum computation capacity can greatly decrease the computation time and more time can be used for transmission, thus reducing the transmit power and total energy. For high computation capacity region, each user has chosen its optimal computation capacity and the increase of maximum computation capacity does not affect the computation capacity allocation result, thus leading to stable energy consumption.

		\section{Conclusions}
		In this paper, the problem of wireless resource allocation and semantic information extraction for energy efficient semantic communications over wireless networks with rate splitting is investigated. In the considered model, the BS first extracts the semantic information from its large-scale data, and then transmits the small-sized semantic information to each user which recovers the original data based on the local common knowledge. In the downlink transmission, the rate splitting scheme is adopted, while the private small-sized semantic information is transmitted through private message and the common knowledge is transmitted through common message. Due to limited wireless resource, both computational energy and transmission energy need to be considered. This joint computation and communication problem is considered as an optimization problem whose goal is to minimize the total energy consumption of the network under both task completion and semantic accuracy constraints. An iterative algorithm is presented to solve this problem, where at each step, the optimal solutions for semantic information extraction ratio  and computation frequency are derived. Numerical results show the effectiveness of the proposed algorithm.

		\bibliographystyle{IEEEtran}
		\bibliography{IEEEabrv,MMM}

	\end{document}